\documentclass{article}

\usepackage[cmex10]{amsmath}
\usepackage{amssymb}
\usepackage{amsthm}
\usepackage{algorithm}
\usepackage{algorithmic}
\usepackage{fixltx2e}

\def\P{\ensuremath{\mathfrak{P}}}
\def\F{\ensuremath{\mathbb{F}}}
\def\N{\ensuremath{\mathbb{N}}}
\def\C{\ensuremath{\mathcal{C}}}
\def\IP{\ensuremath{\mathbb{I}\mathbb{P}}}

\def\equaldef{\ensuremath{\triangleq}}
\def\O{\ensuremath{\mathcal{O}}}

\DeclareMathOperator{\RS}{RS}

 \newtheorem{definition}{Definition}
 \newtheorem{proposition}{Proposition}
 \newtheorem{lemma}{Lemma}
 \newtheorem{corollary}{Corollary}
 \newtheorem{example}{Example}

\newcommand{\parenth}[1]{\ensuremath{\left( #1 \right)}}
\newcommand{\mybrace}[1]{\ensuremath{\left\lbrace #1 \right\rbrace}}
\newcommand{\floor}[1]{\ensuremath{\left\lfloor #1 \right\rfloor}}

\title{Re-encoding reformulation and
  application to Welch-Berlekamp algorithm}
\author{Morgan Barbier\\
  ENSICAEN -- GREYC\\
  \texttt{morgan.barbier@ensicaen.fr}}

\begin{document}
\sloppypar
\maketitle

\begin{abstract}
  The main decoding algorithms for Reed-Solomon codes are based on
  a bivariate interpolation step, which is expensive in
  time complexity. Lot of interpolation methods were proposed in
  order to decrease the complexity of this procedure, but they stay still
  expensive. Then Koetter, Ma and Vardy proposed in 2010 a technique,
  called re-encoding, which allows
  to reduce the practical running time. However, this trick is only devoted for the
  Koetter interpolation algorithm. We propose a reformulation of the
  re-encoding for any interpolation methods. The assumption for this reformulation permits only to apply it to
  the Welch-Berlekamp algorithm.
\end{abstract}

\section{Introduction}
The algebraic decoding algorithms for the Reed-Solomon codes have been deeply
studied for the last decades, especially their decoding
algorithms. The Welch-Berlekamp decoding method provides a simple
approach to decode the Reed-Solomon codes up to the correction
capacity of the code \cite{WelchBerlekamp}. Then in 1997, Sudan generalizes
this approach to decode beyond this bound, which supplies the first
list decoding method for this family \cite{SudanLD}. Two years
latter, Guruswami and Sudan introduced another generalization of the
last method to correct even more errors, that is up to the Johnson's bound
\cite{GSAlgo}. In these three previous methods, a bivariate
interpolation step is needed, moreover their time complexities are given
by this procedure, which is expensive. Thus a lot of algorithms were
proposed to solve the bivariate interpolation as efficient as possible
\cite{Koetter,Alek,GaboritRuatta,AugotZeh,Trifonov,Beelen_Brander}. Even with these
computation improvements, the bivariate interpolation step stays
expensive.\\

In this way, Koetter, Ma, and Vardy introduced the notion of
re-encoding \cite{Reencoding}. This trick does not decrease the
asymptotic complexity, but leads to a considerable gain in
practice. The re-encoding can be split into three phases as following :
start to perform a translation by a codeword on the received word such
that $k$ positions become null, then  modify the intern
statement of the interpolation algorithm to have benefits of the null
positions, finally remove after the interpolation the translation did
at the first step. This technique implies to modify the intern state
of the interpolation algorithm in
relation with the null positions to speed up the running time of the
interpolation step. This adjustment of the intern state of interpolation algorithm is the main,
and maybe the only one, drawback of re-encoding.\\

In this article, we propose a new reformulation of the
re-encoding. This reformulation permits to use the re-encoding trick
with any bivariate interpolation algorithm without preliminary
modification. However to be generic is
under an assumption between the multiplicity and the $Y$-degree of the
interpolated polynomial. We apply this reformulation to the
Welch-Berlekamp algorithm and we observe that the gain is huge.\\

This article is organized in the following way: in Section~\ref{sec:decodingAlgo} we recall the main decoding algorithms
based on interpolation as Welch-Berlekamp, Sudan and
Guruswami-Sudan. Section~\ref{sec:reencoding} is devoted to recall
the principle of the original re-encoding and to introduce our
reformulation. Finally, in Section~\ref{sec:application} we apply
our revisited re-encoding to Welsh-Berlekamp algorithm and present the
performances.

\section{Interpolation based decoding algorithms}
\label{sec:decodingAlgo}
\subsection{Bivariate interpolation for the decoding}
Different decoding algorithms are based on the bivariate
interpolation. This step is the most expensive one, and the asymptotic
complexity is given by this bivariate interpolation. For example,
Welch-Berlekamp, Sudan and Guruswami-Sudan algorithms are based on
this procedure. Since the list decoding
algorithm for alternant codes \cite{ABC_LD_Goppa}, is also based on
interpolation step, we can also apply the re-encoding on it. In
this article we propose to deal only with the decoding algorithms for
Reed-Solomon codes, this is why we propose at first to recall the
definition of this class of codes.
\begin{definition}[Reed-Solomon codes]
  Let $\alpha_1,\dots,\alpha_n \in \F_q$ be $n$ distinct elements of
  $\F_q$. The {\em Reed-Solomon code} of dimension $k$ and support
  $(\alpha_i)$ is given by
  $$
  \RS[\alpha, k] = \mybrace{(P(\alpha_1),\dots,P(\alpha_n)) : P\in
    \F_q[X]_{<k}}.
  $$
\end{definition}
The three following algorithms are based on the same principle:
\begin{enumerate}
\item Compute a bivariate polynomial by interpolation of the received
  word $y$ and the support of the Reed-Solomon code $\alpha$.
\item Compute the univariate polynomial(s) $P$ which generated the
  codeword(s), as $Y$-root(s) of the bivariate polynomial.
\end{enumerate}
The differences between the three following algorithms are the
parameters of the bivariate interpolation, and it represent the most
expensive cost in time complexity of these methods. In the following,
we present quickly the main decoding algorithms for Reed-Solomon
codes, the interested reader can find more information in \cite{GuruswamiThesis} for example.

\subsubsection{Welch-Berlekamp}
The Welch-Berlekamp algorithm is an unambiguous decoding algorithm
devoted to the Reed-Solomon codes \cite{WelchBerlekamp}. Faster unambiguous decoding algorithms exist, as for example extended Euclide or Berlekamp-Massey algorithms, but they are devoted only to cyclic Reed-Solomon codes. Moreover, the
most famous list decoding algorithms are based on this method. This is
why, we propose to recall the main step of this algorithm.\\

This method is based on the computation of the bivariate polynomial by
interpolation satisfying
$$
(\IP_{WB}) \equaldef
\left\lbrace
\begin{array}{l}
   0 \ne Q(X,Y) \equaldef Q_0(X) + Y Q_1(X),\\
  Q(\alpha_i,y_i) = 0,\ \forall i \in \mybrace{1,\dots,n},\\
  \deg Q_0 \le n-t-1,\\
  \deg Q_1 \le n-t-k,
\end{array}
\right.
$$
where $t=\floor{\frac{n-k}{2}}$ is the correction capacity of the
Reed-Solomon code. Thus we obtain the pseudo-code in Algorithm~\ref{algo:WB}.
\begin{algorithm} 
\caption{\label{algo:WB}Welch-Berlekamp}
\begin{algorithmic}
\REQUIRE{The received word $y \in \F_q^n$ and the Reed-Solomon code $\C$.}
\ENSURE{The codeword $c \in \C$ if it exists such that $d(c,y)\le t=\floor{\frac{n-k}{2}}$, under the polynomial form.}
        
 \STATE
 \STATE $Q(X,Y) \gets$ Interpolation($\IP_{WB},\C$)
 \RETURN{$- \frac{Q_0(X)}{Q_1(X)}$.}
\end{algorithmic}
\end{algorithm}

\subsubsection{Sudan}
Sudan realized that if we wish to correct more errors, with the
Welch-Berlekamp algorithm, it could happen that there would exist different
$Y$-roots of the bivariate polynomial satisfying the condition
\cite{SudanLD}. So he proposed to modify the interpolation problem in
this way:
$$
(\IP_{S}) \equaldef
\left\lbrace
\begin{array}{l}
   0 \ne Q(X,Y) \equaldef \sum_{i=0}^{\ell}Q_i(X)Y^i,\\
  Q(\alpha_i,y_i) = 0,\ \forall i \in \mybrace{1,\dots,n},\\
  \deg Q_j \le n-T-1 - j(k-1),\ \forall
  j \in \mybrace{0,\dots,\ell}.
\end{array}
\right.
$$


\subsubsection{Guruswami-Sudan}
The Guruswami-Sudan algorithm introduces the notion of root with
multiplicity from Sudan algorithm \cite{GSAlgo}. Let us to recall the
definition of the Hasse derivative.
\begin{definition}[Hasse derivative]
  Let $Q(X,Y) \in \F_q[X,Y]$ be a bivariate polynomial and $a, b$ be
  two positive integers. The $(a,b)$-th
  {\em Hasse derivative} of $Q$ is
  $$
  Q^{[a,b]}(X,Y) \equaldef \sum_{i=a}^{\deg_X(Q)} \sum_{j=b}^{\deg_Y (Q)}
  {i \choose a} {j \choose b}q_{i,j}X^{i-a}Y^{j-b}.
  $$
\end{definition}
Thanks to the Hasse derivative, we can give the definition of the root
with multiplicity higher than one.
\begin{definition}[Root with multiplicity]
  Let $Q(X,Y) \in \F_q[X,Y]$ be a bivariate polynomial and $(\alpha,\beta)\in
  (\F_q)^2$ be a point. The point $(\alpha, \beta)$ is a {\em root with
  multiplicity} $s\in \N$ if and only if $s$ is the largest integer
  such that for all $i+j<s$
  $$
  Q^{[i,j]}(\alpha,\beta) = 0.
  $$ 
\end{definition}
Guruswami and Sudan noticed that it could happen that for some two
polynomials $P_{i_0}, P_{j_0}$, we have $y_{k_0} =
P_{i_0}(\alpha_{k_0})= P_{j_0}(\alpha_{k_0})$ and so the point
$(\alpha_{k_0},y_{k_0})$ is a root of $Q$ with multiplicity at least
2. So they proposed to add multiplicity constraint during the
bivariate interpolation step.
$$
(\IP_{GS}) \equaldef
\left\lbrace
\begin{array}{l}
   0 \ne Q(X,Y) \equaldef \sum_{i=0}^{\ell}Q_i(X)Y^i,\\
  Q(\alpha_i,y_i) = 0, \text{ with multiplicity }s,\ \forall i \in
  \mybrace{1,\dots,n},\\
  \deg(Q_j) \le s(n-T)-1 - j(k-1),\ \forall
  j \in \mybrace{0,\dots,\ell}.
\end{array}
\right.
$$
thus the pseudo-code of this method is given by:\\
\begin{algorithm}
  \caption{\label{algo:GS}Guruswami-Sudan}
\begin{algorithmic}
  \REQUIRE{The received word $y \in \F_q^n$ and the Reed-Solomon code
  $\C$.}
  \ENSURE{A list of codewords $c_i$ of $\C$, such that
  $\forall i,\ d(y,c_i) \le T$.}

  \STATE 

  \STATE $Q(X,Y) \gets \text{Interpolation}(\IP_{GS},\C)$
  \STATE $(P_1,\dots,P_\ell) \gets \text{Y-Roots(Q(X,Y))}$
  \STATE $Candidate \gets \mybrace{}$
  \FOR{$i \in \mybrace{1,\dots,\ell}$}

  \IF{$d\left(P_i(\alpha),y\right) \le T$}
    \STATE $Candidate \gets Candidate \cup \mybrace{P_i(\alpha)}$
  \ENDIF
  \ENDFOR
  \RETURN{$Candidate$.\\}
\end{algorithmic}
\end{algorithm}

\subsection{Original re-encoding}
\begin{definition}[Interpolation problem]
Let $\P \equaldef \mybrace{(\alpha_1,y_1), \dots, (\alpha_n,y_n)}
\subset \left( \F_q \times \F_q\right)^n$. The {\em interpolation
  problem} with multiplicity $s$ associated to $\P$, $\IP(\P,s)$,
consists in finding $Q(X,Y)$ such that the points $(\alpha_i,y_i)$ are a
root of $Q(X,Y)$ with multiplicity at least $s$.
\end{definition}

\begin{lemma}
  \label{lem:mult}
  Let $s$ be an integer, $\alpha, \beta \in \F_q$ and $Q(X,Y) \in \F_q[X,Y]$ a
  bivariate polynomial such that the point $(\alpha, \beta)$ is a root of $Q$
  with multiplicity $s$. Then for all univariate polynomial $P$ such
  that $P(\alpha)=\beta$, we have
  $$
  (X-\alpha)^s \mid Q(X,P(X)).
  $$
\end{lemma}
\begin{proof}
  See \cite[Lemma 6.6, p. 103]{GuruswamiThesis}.
\end{proof}
We can generalize the previous lemma for all interpolation points,
taking care the multiplicity.
\begin{proposition}
  Let $\P \subset (\F_q \times \F_q)^n$ and $s$ be a positive
  integer. The polynomial $Q(X,Y)$ is a solution of $\IP(\P,s)$ if and
  only if 
  $$
  \forall b \in \mybrace{0,\dots,s-1},\
  \prod_{i=1}^{n}(X-\alpha_i)^{s-b} \mid Q^{[b]}(X,L(X)),
  $$
  where $Q^{[b]}(X,Y) = Q^{[0,b]}(X,Y)$ is the $b$-th Hasse derivative in
  $Y$, and $L(X)$ is the Lagrange polynomial of $\P$, that is for all
  $i \in \mybrace{1,\dots,n}, L(\alpha_i)=y_i$.
\end{proposition}
\begin{proof}
  See \cite[Proposition 1]{AugotZeh}.
\end{proof}
Let $\P \subset (\F_q \times \F_q)^n$ and $L_k(X)$ be the Lagrange
polynomial on $k$ elements of $\P$, without lost in generality,
assuming the $k$ first positions. Let
$$
\P_n = \{(\alpha_i,\underbrace{y_i-L_k(\alpha_i)}_{=r_i}) :
  \forall i \in \mybrace{1,\dots,n}\}.
$$
Then for all $i \in \mybrace{1,\dots,k},\ r_i=0$.
\begin{proposition}
  \label{prop:solwithzeros}
  Let $\P \subset (\F_q \times \F_q)^n$ and $s$ be a positive
  integer. The polynomial $Q(X,Y)$ is a solution of $\IP(\P,s)$ if and
  only if $Q(X,Y+L_k(X))$ is a solution of $\IP(\P_n,s)$.
\end{proposition}
\begin{proof}
  See \cite[Theorem 3]{Reencoding}.
\end{proof}

\section{Revisited re-encoding}
\label{sec:reencoding}
\subsection{Re-encoding and interpolation algorithm}
A problem occurs with the re-encoding process: we have to modify
the interpolation algorithm in order to take care of the $k$ first
interpolation points to speed up the computation. So for each
interpolation algorithm we have to adapt the initialization step to
have the total benefits of the re-encoding step. As far we know, only the
Koetter interpolation algorithm was modified to perform it. Although lot
of interpolation algorithms were proposed, we can use for the moment,
the re-encoding trick only with the Koetter interpolation algorithm.

\subsection{Revisited re-encoding}
Let $L_n(X)$ be the interpolation Lagrange polynomial of the set
$\P_n$. Thus $\forall i \in \mybrace{1, \dots,n},\ L_n(\alpha_i) = r_i$,
and $\deg(L_n) \le n-1$. Since for all $i \in \mybrace{1,\dots,k}\ r_i
=0$, it exists the polynomial $L_{n-k}(X)$ such that
$$
L_{n-k}(X) = \frac{L_{n}(X)}{\prod_{i=1}^{k}(X-\alpha_i)}.
$$
Thanks to the previous remark on the Lagrange polynomials, we deduce
the following proposition which is the key ingredient of our
reformulation.
\begin{proposition}
  \label{prop:remove_zeros}
   Let $\P_{n-k}$ be the point set without zeros defined as
$\P_{n-k} \equaldef \mybrace{(\alpha_{k+1},L_{n-k}(\alpha_{k+1})),
  \cdots, (\alpha_n,L_{n-k}(\alpha_n))} \subset (\F_q \times
\F_q)^{n-k}$,
$$
R(X,Y) = \sum_{j=0}^{\deg_Y(R)} R_j(X) Y^j,
$$
be a bivariate polynomial over $\F_q$ and $s$ be a positive integer
such that $s \ge \deg_Y R$. The polynomial $R$ is a solution of
$\IP(\P_{n-k},s)$ if and only if 
$$
Q(X,Y) = \sum_{j=0}^{\deg_Y(R)} \left(R_j(X)
  \prod_{i=1}^{k}(X-\alpha_i)^{s-j}\right)Y^j,
$$ 
is a solution of $\IP(\P_n,s)$.
\end{proposition}
\begin{proof}
Since $R$ is a solution of $\IP(\P_{n-k},s)$, then for all $b \in \mybrace{0,\dots,s-1}$
{\small
\begin{eqnarray*}
  \prod_{i=k+1}^{n}(X-\alpha_i)^{s-b} & \mid & R^{[b]}(X,L_{n-k}(X))\\
  \prod_{i=k+1}^{n}(X-\alpha_i)^{s-b} & \mid & \sum_{j=b}^{\deg_Y(R)} {j
    \choose b} R_j(X) (L_{n-k}(X))^{j-b}\\
  \prod_{i=1}^{n}(X-\alpha_i)^{s-b} & \mid & \sum_{j=b}^{\deg_Y(R)} {j
    \choose b }\left(R_j(X) \prod_{i=1}^{k}(X-\alpha_i)^{s-j}\right)
  (L_n(X))^{j-b}\\
  \prod_{i=1}^{n}(X-\alpha_i)^{s-b} & \mid & Q(X,L_{n}(X)).
\end{eqnarray*}
}
Since $s \ge \deg_Y R, s-j\ge 0$, the statement is hold.
\end{proof}
\begin{corollary}
  It exists a polynomial $R(X, Y)$ solution of $\IP(\P_{n-k},s)$ with $s \ge \deg_Y R$ if and only if exists $Q(X,Y)$ a solution of $\IP(\P_n,s)$.
\end{corollary}
\begin{proof}
  Since the first and last line of the proof of Proposition~\ref{prop:remove_zeros} are equivalent, the statement is hold.
\end{proof}
Thanks to the Proposition~\ref{prop:solwithzeros}, we can compute a solution
of the interpolation problem on $\P =
\mybrace{(\alpha_1,y_1),\dots,(\alpha_n,y_n)}$ from
$\P_n=\left\lbrace(\alpha_1,0),\dots,(\alpha_k,0),\right.$ $\left.
  (\alpha_{k+1},y_{k+1}-L_k(\alpha_{k+1})),\dots,
  (\alpha_n,y_{n}-L_k(\alpha_n))\right\rbrace$. Our revisited
re-encoding could be seen as a decoding on the puncturing code. Since
the Reed-Solomon code are MDS, the punctured code has the same
dimension and it is also a Reed-Solomon code. We could imagine to
reiterate the re-encoding process taking $\P_{n-k}=\P'$, then the
decoding will make on the multi puncturing code and the correction
radius will decrease.\\

The Proposition~\ref{prop:remove_zeros} is under the assumption that the
multiplicity $s$ is greater or equal than the $Y$-degree of $R$, a
solution of $\IP(\P_{n-k},s)$. Which is not a problem, because a
solution of $\IP(\P_{n-k},s+k)$, for all positive integer $k$, is also
a solution of $\IP(\P_{n-k},s)$. However, this artificial augmentation
of the multiplicity could increase also the $X$-degree of the
solution, and so introduces some issue for the interpolation problem
related to the decoding. This is why we deal only with the
Welch-Berlekamp algorithm in Section~\ref{sec:application}.

\section{Application to the Welch-Berlekamp algorithm}
\label{sec:application}
\subsection{Straightforward application}
In the Welch-Berlekamp decoding context, use the principle of the
revisited re-encoding, is straightforward. Indeed, the only
condition in order to make practical our re-encoding is that
multiplicity $s$ is greater or equal than the $Y$-degree of the
bivariate polynomial to compute. In the Welch-Berlekamp context the
multiplicity $s$ is exactly equal to the $Y$-degree, that is 1. 
Let $S(X, Y) = S_0(X) + Y S_1(X) \in \F_q[X, Y]$ be a solution of
$\IP(\P_{n-k}, 1)$, then $R$ given by the
Proposition~\ref{prop:remove_zeros}:
\begin{eqnarray*}
  R(X, Y) & = & S_0(X) \prod_{i=1}^{k} (X-\alpha_i) + Y S_1(X),
\end{eqnarray*}
is a solution of the $\IP(\P_{n},1)$. Keeping the same notations and
using the Proposition~\ref{prop:solwithzeros}, we deduce directly a solution
of the interpolation problem $\IP(\P, 1)$ from the simpler one
$\IP(\P_{n-k},1)$. Let $Q(X, Y) \in \F_q[X,Y]$ such that 
{\small
\begin{eqnarray*}
  Q(X, Y) & = & R(X, Y + L_k(X))\\
  & = & \parenth{S_0(X) \prod_{i=1}^{k} (X-\alpha_i) + L_k(X) S_1(X)}
  + Y S_1.
\end{eqnarray*}
}
In order to satisfy the interpolation conditions of the
Welch-Berlekamp algorithm, we must have: $\deg S_1 \le n-t -k$ and
$\deg S_0 \le n-t-1-k$. It can be rewritten as
$$
\forall j\in \mybrace{0,1},\ \deg S_j \le n-t-k-1 -j(-1).
$$
We deduce that the weighted-degree changes during the bivariate
interpolation. Using the example described below,
we have to interpolate $n-k$ points with the weighted-degree equal to
-1, instead of interpolating $n$ points with the weighted-degree
$k-1$, without modifying the intern state of the interpolation
algorithm. As already noticed in \cite{Reencoding}, this is not a monomial order since $Y<1$.
Let us illustrate our claim with a toy example.
\begin{example}
  Let $\F_8$, $\alpha$ be a 7-th primitive root of the unity such that
  $\alpha^3 + \alpha + 1 = 0$, $\C$ be the Reed-Solomon code
  $\RS[(\alpha^i)_{i=0,\dots,6}, 2]$ over $\F_8$. Hence the
  Welch-Berlekamp method can correct up to $\floor{\frac{d-1}{2}} = 2$
  errors. Let $P(X) = \alpha^6 X + \alpha^5 \in \F_8[X]$ be the message
  under its polynomial form. The associated codeword is then $(\alpha,
  \alpha^4, \alpha^6, \alpha^3, \alpha^2, 1, 0)$. Assume there are 2
  errors occur during the transmission in the first and 5-th positions,
  the received word is $(\alpha^5, \alpha^4, \alpha^6, \alpha^3,
  \alpha^3, 1, 0)$.\\

  Now let us perform the revisited re-encoding. Using the previous
  notations, the Lagrange interpolation polynomial of the
  original interpolation points set $\P_n = \mybrace{(1,\alpha^5),
    (\alpha,\alpha^4),
    (\alpha^2,\alpha^6),
    (\alpha^3,\alpha^3),
    (\alpha^4,\alpha^3),
    (\alpha^5, 1),
    (\alpha^6, 0)
  }$ is 
  $$
  L_n = X^6 + \alpha^4 X^4 + \alpha^2 X^3 + \alpha^3 X^2 +
  \alpha^2 X + \alpha^2.
  $$
  Assume that we want to vanish the 2 first points, then the Lagrange
  interpolation polynomial on these points is $L_k = \alpha^4 X + 1$,
  and the quotient
  $$
  L_{n-k} = \frac{L_n(X)}{(X-1)(X-\alpha)} = X^4 + \alpha^3 X^3 + X^2
  + X + \alpha.
  $$
  Then the new interpolation points set is 
  $$
    \P_{n-k} = \mybrace{ (\alpha^2, \alpha^4),
    ( \alpha^3, \alpha^2),
    ( \alpha^4, 0 ),
    ( \alpha^5, \alpha^6 ),
    ( \alpha^6, \alpha )}.
  $$
  Hence the bivariate polynomial which interpolates $\P_{n-k}$ with
  multiplicity $s=1$ and weighted-degree -1 is
  $$
  S(X, Y) = Y(\underbrace{\alpha^6 X^2 + \alpha^4 X + \alpha^3}_{ =
      S_1}) + \underbrace{\alpha^2X + \alpha^6}_{=S_0}.
  $$
  We deduce the polynomial which interpolates the $
  \P_{n} = \mybrace{ (1, 0),
      (\alpha, 0 ),
      (\alpha^2, \alpha^4),
    ( \alpha^3, \alpha^2),
    ( \alpha^4, 0 ),
    ( \alpha^5, \alpha^6 ),
    ( \alpha^6, \alpha )},
  $
  is
  {\small
  \begin{eqnarray*}
    R(X, Y) & = & Y S_1 (X) + (X-1)(X-\alpha) S_0(X)\\
     & = & Y(\alpha^6X^2 + \alpha^4 X + \alpha^3) + \alpha^2 X^3 +
     \alpha X^2 + \alpha^5 X + 1.
  \end{eqnarray*}
  }
  To finish the reconstruction step of the interpolation, we compute 
  \begin{eqnarray*}
    Q(X,Y) &= & R(X, Y+L_n(X))\\
    & = & Y(\underbrace{\alpha^6 X^2 + \alpha^4 X + \alpha^3}_{=Q_1})
    + \underbrace{\alpha^5 X^3 + \alpha^6 X^2 + \alpha}_{=Q_0}.
  \end{eqnarray*}
  In the Welch-Berlekamp algorithm the $Y$-root search is
  trivial. Indeed, it
  consists only in the division of the $Q_0$ by $Q_1$
  \begin{eqnarray*}
    P  = & - \frac{Q_0}{Q_1} & = \alpha^6 X + \alpha^5,
  \end{eqnarray*}
  which is exactly the sent message under the polynomial form.
\end{example}

\subsection{Performance}
In this section, we propose to compare the Welch-Berlekamp running times based on different interpolation methods. In Table~\ref{tab:compareS}, we compare the Welch-Berlekamp algorithm based on solving a linear system with no re-encoding, and with our revisited re-encoding. In Table~\ref{tab:compareK}, we compare Welch-Berlekamp algorithm based on Koetter interpolation: without re-encoding, with original re-encoding and with our revisited re-encoding. These experimentations were done on a 2.13GHz Intel(R) Xeon(R). The timings presented in Table~\ref{tab:compareS} and Table~\ref{tab:compareK} are in seconds unit for 100 iterations for each set of parameters. 

\subsubsection{Linear systems for interpolation}
From \cite{gathen}, the asymptotic complexity of solving linear systems is $\O(n^{2.3727})$. This complexity could be discussed in practical context, for our implementation we use the black box solver of \texttt{MAGMA} \cite{magma}. In Table~\ref{tab:compareS}, we compare the running time of Welch-Berlekamp algorithm based on solving linear system equations without re-encoding and with our revisited re-encoding method. We can see that using the revisited re-encoding provides an important gain especially as the code dimension $k$ is large.
\begin{table}[h]
  \centering
 \begin{tabular}{|c|c||c|c|}
    \hline
    $m$ & $\C$ & usual & revisited re-encoding\\
    \hline
    \hline
      4
 & $\RS[ 15 ,  8  ]$ &   0.070   &   0.010  \\
    \cline{2-4}
 & $\RS[ 15 ,  10  ]$ &   0.040   &   0.030  \\
    \cline{2-4}
 & $\RS[ 15 ,  12  ]$ &   0.050   &   0.010  \\
    \cline{2-4}
 & $\RS[ 15 ,  14  ]$ &   0.050   &   0.000  \\
    \hline
    \hline
      5
 & $\RS[ 31 ,  16  ]$ &   0.150   &   0.080  \\
    \cline{2-4}
 & $\RS[ 31 ,  20  ]$ &   0.150   &   0.040  \\
    \cline{2-4}
 & $\RS[ 31 ,  24  ]$ &   0.150   &   0.040  \\
    \cline{2-4}
 & $\RS[ 31 ,  28  ]$ &   0.140   &   0.030  \\
    \hline
    \hline
      6
 & $\RS[ 63 ,  32  ]$ &   0.530   &   0.210  \\
    \cline{2-4}
 & $\RS[ 63 ,  40  ]$ &   0.520   &   0.150  \\
    \cline{2-4}
 & $\RS[ 63 ,  48  ]$ &   0.490   &   0.120  \\
    \cline{2-4}
 & $\RS[ 63 ,  56  ]$ &   0.500   &   0.080  \\
    \hline
    \hline
      7
 & $\RS[ 127 ,  64  ]$ &   2.100   &   0.730  \\
    \cline{2-4}
 & $\RS[ 127 ,  80  ]$ &   2.050   &   0.500  \\
    \cline{2-4}
 & $\RS[ 127 ,  96  ]$ &   1.970   &   0.320  \\
    \cline{2-4}
 & $\RS[ 127 ,  112  ]$ &   1.890   &   0.230  \\
    \hline
    \hline
      8
 & $\RS[ 255 ,  128  ]$ &   9.100   &   2.830  \\
    \cline{2-4}
 & $\RS[ 255 ,  160  ]$ &   8.870   &   1.790  \\
    \cline{2-4}
 & $\RS[ 255 ,  192  ]$ &   8.600   &   1.080  \\
    \cline{2-4}
 & $\RS[ 255 ,  224  ]$ &   8.370   &   0.700  \\
    \hline
  \end{tabular}
  \caption{\label{tab:compareS}
  Comparison between the Welch-Berlekamp using linear system for interpolation without re-encoding and our revisited re-encoding. The shown timings are in second unit for 100 computations.}
\end{table}

\subsubsection{Koetter algorithm for interpolation}
Since it is the main goal of this article, we assume that we
cannot modify the intern state of the interpolation algorithm.
The asymptotic complexity of Koetter algorithm is $\O(L N^2)$; where
$L$ is the $Y$-degree of the bivariate polynomial $Q$ and $N$ the
number of the linear constraints given by the interpolation
conditions. Then the complexity of the standard Welch-Berlekamp
algorithm is $\O(n^2)$. While the complexity of the original
re-encoding is $\O((n-k)^2)$ with inter state modifications and $\O(n^2)$ without, our revisited re-encoding exhibits an asymptotic complexity of $\O((n-k)^2)$ without modification of interpolation method. We propose to compare 3 decoding methods:
Welch-Berlekamp algorithm without re-encoding, Welch-Berlekamp
algorithm with the original re-encoding, and finally the
Welch-Berlekamp with our revisited re-encoding. These 3 decoding
methods were implemented with the same interpolation function, without
modification or particular parameterization. As in the solving linear system equations case, we remark that the revisited re-encoding is both faster than the usual and original re-encoding methods. The gain is important especially as the code dimension $k$ is.
\begin{table}[h]
  \centering
  {\scriptsize
  \begin{tabular}{|c|c||c|c|c|}
    \hline
    $m$ & $\C$ & usual & original re-encoding & revisited
    re-encoding\\
    \hline
    \hline
    4 & $\RS[ 15 ,  8 ]$ &  0.270  &  0.230  &  0.090  \\
    \cline{2-5}
    & $\RS[ 15 ,  10 ]$ &  0.260  &  0.210  &  0.080  \\
    \cline{2-5}
    & $\RS[ 15 ,  12 ]$ &  0.250  &  0.180  &  0.050  \\
    \cline{2-5}
    & $\RS[ 15 ,  14 ]$ &  0.230  &  0.180  &  0.050  \\
    \hline
    \hline
    5 & $\RS[ 31 ,  16 ]$ &  0.930  &  0.760  &  0.250  \\
    \cline{2-5}
    & $\RS[ 31 ,  20 ]$ &  0.820  &  0.710  &  0.220  \\
    \cline{2-5}
    & $\RS[ 31 ,  24 ]$ &  0.820  &  0.660  &  0.100  \\
    \cline{2-5}
    & $\RS[ 31 ,  28 ]$ &  0.840  &  0.550  &  0.080  \\
    \hline
    \hline
    6 & $\RS[ 63 ,  32 ]$ &  3.440  &  3.130  &  1.070  \\
    \cline{2-5}
    & $\RS[ 63 ,  40 ]$ &  3.480  &  2.890  &  0.650  \\
    \cline{2-5}
    & $\RS[ 63 ,  48 ]$ &  3.460  &  2.580  &  0.390  \\
    \cline{2-5}
    & $\RS[ 63 ,  56 ]$ &  3.350  &  2.300  &  0.260  \\
    \hline
    \hline
    7 & $\RS[ 127 ,  64 ]$ &  16.760  &  15.220  &  4.440  \\
    \cline{2-5}
    & $\RS[ 127 ,  80 ]$ &  16.840  &  14.160  &  2.570  \\
    \cline{2-5}
    & $\RS[ 127 ,  96 ]$ &  17.400  &  13.160  &  1.260  \\
    \cline{2-5}
    & $\RS[ 127 ,  112 ]$ &  17.780  &  11.600  &  0.560  \\
    \hline
    \hline
    8 & $\RS[ 255 ,  128 ]$ &  100.780  &  92.070  &  21.150  \\
    \cline{2-5}
    & $\RS[ 255 ,  160 ]$ &  104.100  &  88.200  &  11.300  \\
    \cline{2-5}
    & $\RS[ 255 ,  192 ]$ &  109.840  &  83.910  &  5.110  \\
    \cline{2-5}
    & $\RS[ 255 ,  224 ]$ &  113.550  &  74.440  &  1.950  \\
    \hline
  \end{tabular}
}
  \caption{\label{tab:compareK}
  Comparison between the Welch-Berlekamp using Koetter interpolation without re-encoding, with
  original re-encoding and our revisited re-encoding. The shown timings are
  in second unit for 100 computations.}
\end{table}

\section{Conclusion and perspective}
We introduce a new reformulation of the re-encoding process which
allows to make it usable with any interpolation algorithm. However
the assumption that the multiplicity $s$ is smaller than the
$Y$-degree is the price to be generic. We perform different tests
with the Welch-Berlekamp algorithm showing that our reformulation
provides a very important gain. A very interesting perspective will be to relieve the assumption to apply this reformulation to list-decoding algorithms.

\bibliographystyle{alpha}
\bibliography{biblio}
\end{document}